%% file: main-arxiv.tex
\documentclass[a4paper,cleveref, autoref,thm-restate]{lipics-v2021}

\usepackage{amsmath,amssymb,color,graphicx,tabularx,wrapfig,xspace,yhmath,enumerate}
\usepackage{xcolor,soul}
\usepackage{hyperref}
\usepackage{caption}
\usepackage{subcaption}
\usepackage[noend]{algpseudocode}
\usepackage{algorithm}
\usepackage{thm-restate}

\title{A Coreset for Approximate Furthest-Neighbor Queries in a Simple Polygon}
\author{Mark de Berg}{Department of Mathematics and Computer Science, TU Eindhoven, the Netherlands}{M.T.d.Berg@tue.nl}{https://orcid.org/0000-0001-5770-3784}{}
\author{Leonidas Theocharous}{Department of Mathematics and Computer Science, TU Eindhoven, the Netherlands}{l.theocharous@tue.nl}{}{}
\authorrunning{M.~de Berg and L.~Theocharous}
\Copyright{Mark de Berg and Leonidas Theocharous}

\makeatletter
\newcommand{\multiline}[1]{%
  \begin{tabularx}{\dimexpr\linewidth-\ALG@thistlm}[t]{@{}X@{}}
    #1
  \end{tabularx}
}
\makeatother

\newif\ifComments
\Commentstrue
\ifComments
    \newcommand{\mdb}[1]{\textcolor{blue}{ MdB: #1}}
    \newcommand{\lt}[1]{\textcolor{red}{ LT: #1}}
\else
    \newcommand{\mdb}[1]{}
    \newcommand{\lt}[1]{}
\fi

\newcommand{\BeginMyItemize}{\begin{itemize}\setlength{\itemsep}{-\parskip}}
\newcommand{\EndMyItemize}{\end{itemize}}

\newcommand{\BeginMyEnumerate}{\begin{enumerate}\setlength{\itemsep}{-\parskip}}
\newcommand{\EndMyEnumerate}{\end{enumerate}}

\renewcommand{\leq}{\leqslant}
\renewcommand{\geq}{\geqslant}

\newtheorem{defin}{Definition}
\newtheorem{lem}[defin]{Lemma}
\newtheorem{myfact}[defin]{Fact}

\newenvironment{myquote}%
  {\list{}{\leftmargin=4mm\rightmargin=4mm}\item[]}%
  {\endlist}
\newcommand{\claiminproof}[2]{\begin{myquote}\noindent\emph{Claim~#1.} #2 \end{myquote}}
\newenvironment{proofinproof}{\begin{myquote}\noindent\emph{Proof.}}{\hfill $\lhd$ \end{myquote}}

\newcommand{\Reals}{\mathbb{R}}

\newcommand{\bd}{\partial}

\newcommand{\diam}{\mathrm{diam}}

\newcommand{\tree}{\ensuremath{\mathcal{T}}}

\newcommand{\T}{\ensuremath{\mathcal{T}}}
\newcommand{\calP}{\ensuremath{\mathcal{P}}}
\renewcommand{\P}{\calP}

\newcommand{\ceil}[1]{\left\lceil #1 \right\rceil}

\newcommand{\eps}{\varepsilon}

\newcommand{\etal}{\emph{et al.}\xspace}

\newcommand{\coreset}{C}

\newcommand{\Span}{\mathrm{span}}
\newcommand{\ext}{\mathrm{ext}}
\newcommand{\rch}{\mbox{\sc rch}}
\newcommand{\fpvd}{\mbox{\sc fpvd}}
\newcommand{\fn}{\mbox{\sc fn}}
\newcommand{\dir}{\mathrm{dir}}
\newcommand{\first}{\mathrm{first}}
\newcommand{\last}{\mathrm{last}}
\newcommand{\Hone}{\P_1(\Gamma^*)}
\newcommand{\Htwo}{\P_2(\Gamma^*)}
\newcommand{\fext}{\overrightarrow{\ext}}
\newcommand{\bext}{\overleftarrow{\ext}}
\newcommand{\bfext}{\overleftrightarrow{\ext}}

\category{} 

\relatedversion{} 

\funding{MdB and LT are supported by the  Dutch Research Council (NWO) through Gravitation-grant NETWORKS-024.002.003.}

\nolinenumbers 

\hideLIPIcs  

\EventEditors{John Q. Open and Joan R. Access}
\EventNoEds{2}
\EventLongTitle{42nd Conference on Very Important Topics (CVIT 2016)}
\EventShortTitle{CVIT 2016}
\EventAcronym{CVIT}
\EventYear{2016}
\EventDate{December 24--27, 2016}
\EventLocation{Little Whinging, United Kingdom}
\EventLogo{}
\SeriesVolume{42}
\ArticleNo{23}

\begin{document}
\thispagestyle{empty}
\maketitle

\begin{abstract}
Let $\P$ be a simple polygon with $m$ vertices and let $P$ be a set of $n$ points inside~$\P$.
We prove that there exists, for any $\eps>0$, a set $\coreset \subset P$ of size $O(1/\eps^2)$ such that
the following holds: for any query point $q$ inside the polygon~$\P$, the geodesic distance from $q$ to its furthest neighbor in $\coreset$
is at least $1-\eps$ times the geodesic distance to its further neighbor in~$P$. Thus the set~$\coreset$
can be used for answering $\eps$-approximate furthest-neighbor queries with a data structure
whose storage requirement is independent of the size of~$P$. The coreset can be constructed in 
$O\left(\frac{1}{\eps} \left( n\log(1/\eps) + (n+m)\log(n+m)\right) \right)$ time.
\end{abstract}

\input{1-introduction}

\input{2-coreset}
\input{3-preprocessing}

\input{4-conclusion}

\bibliography{refs}

\newpage

\appendix
\section{Appendix}
\input{app-triangle-proof}
\end{document}

%% file: 1-introduction.tex
\section{Introduction}

\subparagraph{Background and problem statement.}
In nearest-neighbor searching the goal is to preprocess a set $P$ of $n$ points in a metric space
such that, for a query point~$q$, one can quickly retrieve its nearest neighbor in~$P$.
Exact and approximate data structures for nearest-neighbor queries have been studied 
in many settings, including general metric spaces, 2- and higher-dimensional 
Euclidean spaces, and inside (simple) polygons.  

\emph{Furthest-neighbor queries}, while not as popular as nearest-neighbor queries,
have been studied extensively as well. Here the goal is to report a point in $P$ that 
\emph{maximizes} its distance to the query point~$q$. We will denote 
an (arbitrarily chosen, if it is not unique) furthest neighbor of $q$ in $P$ by $\fn(q,P)$.
Furthest-neighbor queries can be answered exactly by performing point location in
$\fpvd(P)$, the furthest-point Voronoi diagram of~$P$. 
In $\Reals^2$, $\fpvd(P)$ has complexity $O(n)$ and it can be computed
in $O(n \log n)$ time~\cite{DBLP:books/lib/BergCKO08}, and even in $O(n)$ time 
if the convex hull is given~\cite{DBLP:journals/dcg/AronovFW93}. 
Furthermore, point-location queries can be answered in $O(\log n)$ time with a 
linear-space data structure~\cite{DBLP:books/lib/BergCKO08},
giving an optimal solution for exact furthest-neighbor queries.
Furthest-site Voronoi diagrams in the plane have also been studied for non-point sites, 
such as line segments~\cite{DBLP:journals/ipl/AurenhammerDK06} and polygons~\cite{DBLP:journals/comgeo/CheongEGGHLLN11},
and even in an abstract setting~\cite{DBLP:journals/ijcga/MehlhornMR01}.

In $\Reals^d$ the worst-case complexity of the furthest-point Voronoi diagram 
is~$O(n^{\ceil{d/2}})$~\cite{DBLP:conf/compgeom/Seidel87}.
(This also holds for polyhedral distance functions~\cite{DBLP:conf/isaac/AurenhammerPS21}.)
Hence, answering furthest-neighbor queries using $\fpvd(P)$
leads to a data structure whose storage is exponential in the dimension~$d$.
Bespamyatnikh~\cite{DBLP:conf/cccg/Bespamyatnikh96a} therefore  studied  
\emph{$\eps$-approximate furthest-neighbor queries}, where the goal is to
report a point in $P$ whose distance to the query point~$q$ is at least $1-\eps$ 
times the distance from~$q$ to its actual furthest neighbor.
He presented a data structure with $O(dn)$ storage that answers $\eps$-approximate 
furthest-neighbor queries in $O(1/\eps^{d-1})$ time.
Note that the storage is linear in~$d$ (and $n$), and that the query time is independent of~$n$.
The query time is exponential in $d$, however, so in very high dimensions
this will not be efficient. Pagh~\etal~\cite{DBLP:journals/is/PaghSSS17}, building on earlier work
of Indyk~\cite{DBLP:conf/soda/Indyk03}, presented a data structure for which both the storage and
the query time are linear in~$d$: for a given~$c>1$ they
can build a data structure that uses $\tilde{O}(dn^{2/c^2})$ storage and 
can report a $c$-approximate\footnote{That is, a point $p\in P$ whose
distance to the query point~$q$ is at least $1/c$ times the distance from $q$ to its actual furthest neighbor.
Note the slight inconsistency in notation when compared to the definition of $\eps$-approximate nearest neighbor.}
furthest neighbor with constant probability in time~$\tilde{O}(dn^{2/c^2})$.
Finally, Goel, Indyk, and Varadarajan~\cite{DBLP:conf/soda/GoelIV01} observed that 
a set $\coreset\subset P$ of at most $d+1$ points defining the smallest enclosing ball of $P$
has the property that, for any query point~$q\in \Reals^d$, 
the furthest neighbor of~$q$
in $\coreset$ is a $\sqrt{2}$-approximation to the actual furthest neighbor.

\subparagraph{Our contribution.}
In this paper we study furthest-neighbor queries on a set $P$ of $n$ points
inside a simple polygon~$\P$ with $m$ vertices, under the geodesic distance.
The \emph{geodesic distance} between two points $x,y \in \P$ 
is the Euclidean length of the shortest path $\pi(x,y)$ connecting them in $\P$.
(Shortest-path problems inside a simple polygon have received considerable attention,
and it is beyond the scope of our paper discuss all the work that has been done.
The interested reader may consult the survey by Mitchell~\cite{DBLP:books/el/00/Mitchell00}
for an introduction into the area.)
It is known that $\fpvd_{\P}(P)$, the geodesic furthest-point Voronoi diagram of $P$ 
inside $\P$, has complexity $O(n+m)$. 
Wang \cite{Wang2023} recently showed that  $\fpvd_{\P}(P)$ can be constructed in optimal $O(m+ n\log n)$ time,
improving an earlier result of Oh, Barba, and Ahn~\cite{DBLP:journals/algorithmica/OhBA20}.

Since the complexity of $\fpvd_{\P}(P)$ is $O(n+m)$
and point-location queries in the plane can be answered in logarithmic time
with a linear-size data structure, the
approach of answering furthest-neighbor queries using $\fpvd_{\P}(P)$ is
optimal with respect to storage and query time. This only holds
for \emph{exact} queries, however: in this paper we show that
$\eps$-approximate furthest-neighbor queries can be answered
using an amount of storage that is \emph{independent of the size of~$P$}. More precisely, 
 we show that there exists a set $\coreset \subset P$ of $O(1/\eps^2)$ points
such that for any query point $q\in \P$, the geodesic furthest neighbor of $q$
among the points in~$\coreset$ is a $(1-\eps)$-approximation of its actual furthest neighbor.
Note that the size of the coreset is not only independent of the size of~$P$,
but also of the complexity of the polygon~$\P$.
We call a set~$\coreset$ with this property an \emph{$\eps$-coreset for furthest-neighbor
queries}. Using such an $\eps$-coreset we can answer $\eps$-approximate 
furthest neighbor queries with a data structure---a point-location structure on 
the geodesic furthest-point Voronoi diagram of $\coreset$---whose storage
is~$O(m+1/\eps^2)$. Note that a similar result is not possible for 
\emph{nearest}-neighbor queries, since for any point $p\in P$
there is a query point~$q$ such that $p$ is the only valid answer to an $\eps$-approximate
nearest-neighbor query (even when $\eps$ is very large).

\subparagraph{Related work.}
To the best of our knowledge, $\eps$-coresets for furthest-neighbor queries have
not been studied explicitly so far, although the observation by 
Goel, Indyk, and Varadarajan~\cite{DBLP:conf/soda/GoelIV01} mentioned above 
can be restated in terms of coresets: any set $P$ in $\Reals^d$ has a $(1/\sqrt{2})$-coreset
of size~$d+1$. It is easy to see that any set $P$ in a space of 
constant doubling dimension~$d$ has a coreset of size $O(1/\eps^d)$. 
Indeed, if $r$ denotes the radius of the smallest enclosing ball of~$P$,
then the distance from any query point $q$ to it furthest neighbor in $P$
is at least~$r$. Hence, covering $B$ with $O(1/\eps^d)$ balls of radius~$\eps \cdot r$
and picking an arbitrary point from $P$ in each non-empty ball, gives an $\eps$-coreset 
for furthest-neighbor queries. 

In $\Reals^d$ this can be improved using $\eps$-kernels~\cite{DBLP:journals/jacm/AgarwalHV04}.
An \emph{$\eps$-kernel} of a point set $P$ in $\Reals^d$ is a subset $\coreset \subset P$
that approximates the width of $P$ in any direction. More precisely, if $\mathrm{width}(\phi,P)$
denotes the width of $P$ in direction $\phi$, then we require that 
$\mathrm{width}(\phi,\coreset) \geq (1-\eps)\cdot \mathrm{width}(\phi,P)$
for any direction~$\phi$. It is easily checked that an $(\eps/2)$-kernel
for $P$ is, in fact, an $\eps$-coreset for furthest-neighbor queries.
Since there always exists an $\eps$-kernel of size~$O((1/\eps)^{(d-1)/2})$~\cite{DBLP:journals/jacm/AgarwalHV04}, 
we immediately obtain an $\eps$-coreset for furthest-neighbor queries in $\Reals^d$ of 
the same size. 

\subparagraph{Overview of our technique.}
While obtaining a small coreset for furthest-neighbor queries is easy under the
Euclidean distance, it becomes much more challenging when $P$ is
a point set in a simple polygon~$\P$ and we consider the geodesic distances. 
It is known (and easy to see) that $\fn(q,P)$ is always a vertex of $\rch(P)$,
the \emph{relative convex hull}\footnote{The relative convex hull~\cite{toussaint-rch} of a set $P$ of
points inside a polygon~$\P$ is the intersection of all relatively convex sets (that is,
all sets $X$ such that for any two points $x,y\in X$ we have $\pi(x,y)\subset X$) containing~$P$;
see Figure~\ref{fig:definitions}(i). Obviously the relative convex hull not only
depends on $P$ but also on the polygon~$\P$, but to avoid cluttering the notation we
simply write $\rch(P)$ instead of $\rch_{\P}(P)$.} of $P$ inside $\P$. Thus, one may be tempted
to cover the boundary of $\rch(P)$ by geodesic balls of radius~$(\eps/2)\cdot \diam(P)$,
where $\diam(P)$ denotes the diameter of~$P$ under the geodesic distance,
and pick a point from~$P$ in any non-empty ball. This will give an $\eps$-coreset
since the geodesic distance from $q$ to $\fn(q,P)$ is at least~$\diam(P)/2$.
Unfortunately, the geodesic distance does not have bounded doubling
dimension, and the boundary of $\rch(P)$ can be arbitrarily much
longer than $\diam(P)$, the geodesic diameter of $P$ inside $\P$.
To obtain a coreset size that depends only on $\eps$, we must therefore use a completely different approach.
\medskip

The idea behind our coreset construction is as follows. 
Consider two points $p_1,p_2\in P$ that define~$\diam(P)$ 
and let $\Gamma := \pi(p_1,p_2)$ be the shortest path between them. 
We will prove that the following property holds for any
query point~$q$ inside $\rch(P)$: either there is a point~$q^*\in \Gamma$ with $\fn(q^*,P)=\fn(q,P)$,
or $\pi(q,\fn(q,P))$ crosses~$\Gamma$. 

Query points of the former type are easy to handle:
we partition~$\Gamma$ into $O(1/\eps)$ pieces $\Gamma_i$ of length $\eps\cdot \diam(P)$,
and we add one furthest neighbor for each piece $\Gamma_i$ into our coreset.

Query points of the latter type are harder to handle. Let $\overline{q}$
denote the point where $\pi(q,\fn(q,P))$ crosses~$\Gamma$. The difficulty
lies in the fact that $\fn(\overline{q},P)$ need not be the same as $\fn(q,P)$
and, in fact, the geodesic distance from $q$ to $\fn(\overline{q},P)$ need not even 
approximate the geodesic distance from $q$ to~$\fn(q,P)$ well enough. The key insight
is that we \emph{can} get a good approximation if, when computing the furthest
neighbor of $\overline{q}$, we restrict our attention to points $p\in P$ such that
the angle between the first link of $\pi(\overline{q},p)$ and
the first link of $\pi(\overline{q},\fn(q,P))$ is at most~$\eps$. Thus we will
define $O(1/\eps)$ ``cones of directions'' and, for each
piece~$\Gamma_i$ we will add one ``directed furthest neighbor'' for each such cone.
The total size of the coreset created in this manner will be $O(1/\eps^2)$.

The above does not immediately work for query points $q$ that lie outside~$\rch(P)$,
because for such query points the path $\pi(q,\fn(q,P))$ may not cross~$\Gamma$.
To overcome this, we show that there exists a set $B$ of $O(1)$ segments, each of length
at most $2\|\Gamma\|$, such that the following property holds for any query point $q\notin \rch(P)$: either $\pi(q,\fn(q,P))$ crosses $\Gamma$ or $\pi(q,\fn(q,P))$ crosses a segment in $B$. We can then split $B$ in $O(1/\eps)$ pieces and then use the ``cone-approach'' for each piece. 

%% file: 2-coreset.tex
\section{A coreset of furthest-neighbor queries in a simple polygon}
In this section we will prove our main theorem, which is as follows.
\begin{theorem}\label{thm:coreset}
For any $0<\varepsilon \leq 1$, there exists an $\eps$-coreset $\coreset\subset P$ of size $O(1/\varepsilon^2)$
for the furthest-neighbor problem, that is, a set $\coreset$
such that for any query point $q \in \P$ we have that $\|\pi(q,\fn(q,\coreset))\|\geq (1-\varepsilon) \|\pi(q,\fn(q,P))\|$.
The coreset $\coreset$ can be constructed in $O\left(\frac{1}{\eps} \left( n\log(1/\eps) + (n+m)\log(n+m)\right) \right)$ time.
\end{theorem}
For now on, we will use $\fn(q)$ as a shorthand for $\fn(q,P)$, the actual nearest neighbor of $q$.
We will first prove the existence of an $\eps$-coreset for query points that lie inside $\rch(P)$,
then extend the result to query points outside $\rch(P)$, and then present an efficient
algorithm to compute the coreset. Before we start we need to introduce some notation and terminology,

\subparagraph{Notation and terminology.}
Recall that $\pi(x,y)$ denotes the shortest path between  two points $x,y\in\P$. 
We assume that $\pi(x,y)$ is directed from $x$ to $y$ and we direct the edges of $\pi(x,y)$ accordingly. 
We denote the first edge of $\pi(x,y)$ by $\first(\pi(x,y))$, while 
the last edge of $\pi(x,y)$ is denoted by~$\last(\pi(x,y))$; 
see Figure~\ref{fig:definitions}(i),
where this is illustrated for~$\Gamma$, the path between the points $p_1$ and $p_2$ defining the diameter.
We define the 
\emph{backwards extension} of a path~$\pi$, denoted by  $\bext(\pi)$, to be the path obtained by extending the first edge of $\pi$ in backwards direction until it hits~$\bd \P$. When we want to extend only 
over a given distance~$\ell>0$, then we write $\bext(\pi,\ell)$. Note that if the extension
would hit $\bd\P$ before length~$\ell$ is reached, then $\bext(\pi,\ell) = \bext(\pi)$.
The \emph{forward extensions} $\fext(\pi)$ and $\fext(\pi,\ell)$ are defined similarly, except that we extend the last edge of $\pi$ in the forward direction. Finally, we define $\bfext(\pi) := \bext(\pi)\cup\fext(\pi)$ and $\bfext(\pi,\ell) := \bext(\pi,\ell)\cup\fext(\pi,\ell)$
to be the extensions of $\pi$ in both directions.

As is standard in the literature, we will use the 1-sphere  $\mathbb{S}^1$ to represent the space of all directions in $\Reals^2$. Thus, for a directed line segment $s\subset \P$, its \emph{direction} will be a point on~$\mathbb{S}^1$, which we will denote by $\dir(s)$. For a shortest path $\pi(x,y)\subset \P$ we define its direction to be $\dir(x,y):= \dir(\first(\pi(x,y))$.  For a point $\phi\in \mathbb{S}^1$, we will denote by $\phi^*$ its antipodal point.
For two points $\phi,\phi'\in \mathbb{S}^1$ that are not anipodal, we denote by $\mathbb{S}^1[\phi,\phi']$ the smaller of the two arcs with endpoints $\phi,\phi'$; see Figure~\ref{fig:definitions}(ii) for an example.

Recall that $\diam(P)$ denotes the (geodesic) diameter of the input point set $P$ inside~$\P$.
We fix two points $p_1,p_2\in P$ that define~$\diam(P)$, so that $\Gamma=\pi(p_1,p_2)$. 
With a slight abuse of terminology we will use the term \emph{diameter} both to
refer to $\Gamma$ (the actual path between $p_1,p_2$) and to $\diam(P)$ (the length of the path).  
We define  $\Gamma^* := \bfext(\Gamma)$ to be the extension of $\Gamma$ 
in both directions until $\partial \P$ is hit. Note that $\Gamma^*$ splits
$\P$ into two regions, one on either side. We denote these two regions by
$\Hone$ and $\Htwo$. Note that $\Hone$---and, similarly, $\Htwo$---need not be 
be a simple polygon, because its interior need not be connected.

\begin{figure}[b]
\begin{center}
\includegraphics{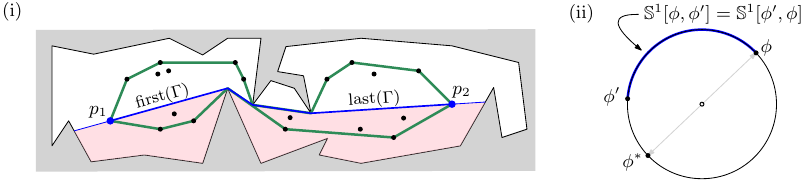}
\end{center}
\caption{(i) The boundary of $\rch(P)$, the relative convex hull,  is shown in green. 
             The points $p_1,p_2$ define $\diam(P)$, so $\Gamma=\pi(p_1,p_2)$. 
             The extension of $\Gamma$ splits~$\P$ into two pieces: $\Hone$, shown in pink,
             and $\Htwo$, shown in white.
             (ii) The 1-sphere of directions~$\mathbb{S}^1$. } \label{fig:definitions}
\end{figure}

\subsection{Reduction to query points on the diameter}
We start with some important properties which concern geodesic triangles. In the following, we denote 
by ${\triangle_\pi}{p'q'r'}$ the geodesic triangle with vertices $p',q',r'\in\P$. Let $p$ denote the 
vertex where $\pi(p',q')$ and $\pi(p',r')$ split. Similarly let $q$ denote the vertex where $\pi(q,r)$
and $\pi(q',p')$ split, and let $r$ denote the vertex at which $\pi(r',p')$ and $\pi(r',q')$ split; see Figure~\ref{fig:intersect}. 
Note that $\triangle_\pi{pqr}$ is also a geodesic triangle. 
In other words, it is a (polygonal) \emph{pseudo-triangle}~\cite{DBLP:journals/ijcga/PocchiolaV96}: a simple polygon with exactly three convex vertices, called the \emph{corners} of the pseudo-triangle.
We will need the following property, illustrated in Figure~\ref{fig:intersect}(i).
\begin{figure}
\begin{center}
\includegraphics{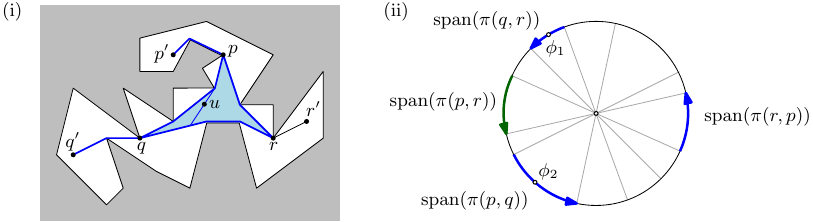}
\end{center}
\caption{(i) A geodesic triangle $\Delta_{\pi}p'q'r'$ and the pseudo-triangle $\Delta_{\pi}pqr$ defined by the points where the paths meet. The extension of $\last(\pi(p,u))$ hits $\pi(q,r)$.
 (ii) Illustration for Lemma~\ref{lem:geo-triangle}. }
\label{fig:intersect}
\end{figure}
\begin{observation} \label{triangle-intersect}
    Let $\triangle_\pi{pqr}$ be a pseudo-triangle. Then for all $u\in \triangle_\pi{pqr}$ such that $u\notin \pi(q,r)$ we have that the forward extension of $\last(\pi(p,u))$ will intersect $\pi(q,r)$. 
\end{observation}
Now consider a path $\pi$. 
Let $e_1,\ldots,e_k$ be the ordered sequence edges of $\pi$, and let $\phi_1,\ldots, \phi_k$
be their directions. We define the \emph{span} of $\pi$ as
$\Span(\pi) := \bigcup_{i=1}^{k-1} \mathbb{S}^1[\phi_i,\phi_{i+1}]$.
In other words, $\Span(\pi)$ contains the directions into which we move
while traversing~$\pi$, where the ``directions over which we turn'' at each vertex are included. By definition, $Span(\pi)$ is connected,  since any two of its consecutive arcs on $\mathbb{S}^1$ share an endpoint.
%
The following lemma provides important properties about the sets of directions 
that occur in a pseudo-triangle.  It is probably known, but since we could not find a reference
we provide a proof in the appendix.
\begin{restatable}{lemma}{geotriangle} \label{lem:geo-triangle}
Let $\triangle_\pi{pqr}$ be a pseudo-triangle. Then 
\begin{enumerate}[(i)] 
   \item  the sets $\Span(\pi(p,q)), \Span(\pi(q,r)), \Span(\pi(r,p))$ are pairwise disjoint and their union contains no antipodal directions.
   \item for any $\phi_1\in \Span(\pi(p,q))$ and $\phi_2\in\Span(\pi(q,r))$, we have  $\Span(\pi(p,r))\subset\mathbb{S}^1[\phi_1,\phi_2]$.
\end{enumerate}
\end{restatable}
\medskip
\noindent Finally, we need the following observation about furthest neighbors. It states that
we extend the first edge of $\pi(q,\fn(q))$ backwards, and then move $q$ along this extension,
then its furthest neighbor does not change.
\begin{observation}\label{obs:walk-back}
For any $q^* \in \bext(\pi(q,\fn(q))$ we have $\fn(q^*)=\fn(q)$. 
\end{observation}
\begin{proof}
Intuitively, this is true because when we move $q$ to $q^*$ along $\bext(\pi(q,\fn(q)))$, its distance to $\fn(q)$ experiences the maximal possible increase (which is the full length of $\overline{q^*q}$). Formally, the observation follows if we can show that $\pi(q^*,\fn(q)) = \overline{q^*q} \cup \pi(q,\fn(q))$, since then $\|\pi(q^*,p)\| \leq \|\overline{q^*q}\| + \|\pi(q,p)\|\leq  \|\overline{q^*q}\| + \|\pi(q,\fn(q))\| = \|\pi(q^*,\fn(q)\| $ for all $p\in P$. It is known that any path that lies in $\P$ is shortest if and only if all its internal vertices bend around reflex vertices of $\P$~\cite[Lemma 1]{alkema_et_al}. Since the paths $\overline{q^*q} \cup \pi(q,\fn(q))$ and $\pi(q,\fn(q))$ have the same set of internal vertices, this concludes the proof.
\end{proof}
We can now prove the lemma that forms the basis of our coreset construction. 
\begin{lemma}\label{lem:inside-rch}
Let $q\in \P$ be a query point. Then either $\pi(q,\fn(q))$ intersects $\Gamma^*$ or there is a point~$q^*\in \Gamma$ with $\fn(q^*)=\fn(q)$. Moreover, if $q\notin \rch(P)$ then 
$\pi(q,\fn(q))$ always intersects $\Gamma^*$.
\end{lemma}
\begin{proof}
If $q\in \Gamma$, then clearly the lemma holds, so we can assume that $q\notin \Gamma$. We now will assume that $\pi(q,\fn(q))$ does not intersect $\Gamma^*$ and show that then there is a point~$q^*\in \Gamma$ with $\fn(q^*)=\fn(q)$. For the rest of the proof, let $p=\fn(q)$ and 
assume without loss of generality that $q$ lies in~$\Hone$. 

Since $\pi(q,\fn(q))$ does not intersect $\Gamma^*$, we have  $p\in \Hone$. Let $u$ be the point where $\pi(p,p_1)$ and $\pi(p,p_2)$ split, let $v$ be the point where $\pi(p_1,p_2)$ and $\pi(p_1,p)$ split and $w$ the point where $\pi(p_2,p_1)$ and $\pi(p_2,p)$ split. Note that $p \in \bd \rch(P)$ and also that the angle formed between $\first(\pi(p,p_1)),\first(\pi(p,p_2))$ within the geodesic triangle $\triangle_\pi{pp_1p_2}$  is convex. Now consider $\fext\left(\pi(p_1,p)\right)$ and $\fext\left(\pi(p_2,p)\right)$ and let $\ell_1=\fext\left(\pi(p_1,p)\right)\setminus \pi(p_1,p)$ and $\ell_2=\fext(\pi(p_2,p))\setminus \pi(p_2,p)$. Also, let $p_1',p_2'$ denote the endpoints of $\Gamma^*$. We partition $\Hone$ into at most four disjoint  non-empty regions, as illustrated in Figure~\ref{fig:lemma-5}.
\begin{figure}
\begin{center}
\includegraphics{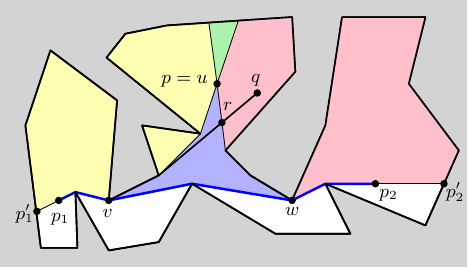}
\end{center}
\caption{The yellow region corresponds to $R_1$, the blue region to $R_2$, the red region to $R_3$, and the green region to $R_4$. In this example, the query point~$q$ lies in $R_3$, so $\pi(q,p_1)$ intersects $\pi(p,p_2)$, at a point~$r$. }
\label{fig:lemma-5}
\end{figure}
\begin{itemize}
    \item$R_1$ corresponds to the region bounded by $\pi(p_1',p)\cup \ell_2$ and $\partial \P$,
    \item  $R_2$ is the region enclosed by $\triangle_\pi{uvw}$,
    \item $R_3$ corresponds to the region bounded by $\pi(p_2',p)\cup \ell_1$ and $\partial \P$,
    \item $R_4$ is the region bounded by $\ell_1\cup \ell_2$ and $\partial \P$.
\end{itemize}
   If $u \neq p$ then $R_4$ is empty, therefore it suffices to consider the case $u= p$. Then we have the following subcases:
\begin{itemize}
    \item Case~I: $q\in R_2$. 
    Notice that $q \in \triangle_\pi{uvw}$ and $q$ does not belong to $\pi(v,w)$. Therefore by Lemma \ref{triangle-intersect}, if $\pi(q,p)$ is extended to the direction of $q$, it will intersect $\pi(v,w)\subset \Gamma$ at some point~$q^*$. By Observation~\ref{obs:walk-back},
    this point $q^*$ has the desired property.  
    \item Case~II: $q\in R_1\cup R_3$. Assume without loss of generality that $q\in R_3$; the case where $q\in R_1$ is symmetrical. Now $\pi(q,p_1)$ has to intersect $\pi(p,p_2)$. The reason is that $R_3$ is bounded by parts of shortest paths that have $p_2$ as an endpoint. Assume that $\pi(q,p)$ crosses $\pi(p,p_2)$ at point $r$. We have that 
    \[
    \|\pi(q,p)\| >\|\pi(q,p_1)\|
    \]
    since $p=\fn(q)$ and we may assume that $ \|\pi(q,p)\| \neq \|\pi(q,p_1)\|$ as otherwise we could use $p_1$ as $\fn(q)$. Also
    \[
    \|\pi(p_1,p_2)\|\geq \|\pi(p,p_2)\|
    \]
    since  $\| \pi(p_1,p_2) \| = \diam(P)$.
   Adding the above inequalities and using that $r$ lies on the paths $\pi(q,p_1)$ and $\pi(p,p_2)$ we get
    \[
    \begin{array}{lll}
        \|\pi(q,p)\|+\|\pi(p_1,p_2)\| & > & \|\pi(q,r)\|+\|\pi(r,p_1)\|+\|\pi(p,r)\|+\|\pi(r,p_2)\| \\ 
          & > & \left(\|\pi(q,r)\|+\|\pi(r,p)\|\right)+ \left(\|\pi(p_1,r)\|+\|\pi(r,p_2)\|\right),
    \end{array}
    \]
   which is a contradiction, because from the triangle inequality we have that $\|\pi(q,p)\|\leq\|\pi(q,r)\|+\|\pi(r,p)\|$ and $\|\pi(p_1,p_2)\|\leq\|\pi(p_1,r)\|+\|\pi(r,p_2)\|$.
   
   \item $p\in R_4$: since the angle formed between $\first(\pi(p,p_1))$ and $\first(\pi(p,p_2))$ is convex, then either the angle between $\first(\pi(p,q))$ and $\first(\pi(p,p_1))$ is obtuse, or the angle between $\first(\pi(p,q))$ and $\first(\pi(p,p_2))$ is obtuse. Then by known properties of geodesic triangles \cite{Pollack1989}, we have that $\|\pi(q,p)\|<\|\pi(q,p_1)\|$ or $\|\pi(q,p)\|<\|\pi(q,p_2)\|$ respectively, which contradicts that $p=\fn(q)$.
\end{itemize}
So, in all cases we get a contradiction, except for Case~I where we have shown the
existence of a point~$q^*\in \Gamma$ with $\fn(q^*)=\fn(q)$. Note that in Case~I
we have $q\in R_2$. Since $R_2\subset \rch(P)$, this also shows that if $q\notin \rch(P)$, then 
$\pi(q,\fn(q))$ has to intersect $\Gamma^*$.
\end{proof}

\subsection{The coreset construction and its analysis}
Let $k=\ceil{4\pi/\eps}$ and let $\phi_1,\phi_2,...,\phi_k$ denote $k$
equally spaced points on $\mathbb{S}^1$, which we will refer to as \emph{canonical directions}.
Note that $|\phi_{i+1}-\phi_i| \leq \eps/2$ for all $i$.
We define $\Phi_j:=\mathbb{S}^1[\phi_j,\phi_{j+1}]$ and call it a 
\emph{canonical cone (of directions)}.
We split $\Gamma$ into $\ell := \ceil{6/\eps}$ equal-length pieces, 
$\Gamma_1,\Gamma_2,...,\Gamma_{\ell}$, so that each piece $\Gamma_i$ has length
at most $(\eps/6)\cdot \|\Gamma\|$.
We denote the endpoints of a piece $\Gamma_i$ by $\gamma_i$ and $\gamma_{i+1}$, 
so $\Gamma_i = \Gamma[\gamma_i,\gamma_{i+1}]$. 

A \emph{pocket of $\Gamma_i$ in $\Hone$} is defined to be a simple polygonal 
region $R\subset \Hone$ whose boundary consists of a portion of $\bd \Hone$ and a 
portion of~$\Gamma_i$. Note that the points where these two portions meet must 
be reflex vertices of $\bd\P$ whose incident edges belong to $\bd \Hone$. 
A pocket of $\Gamma_i$ in $\Htwo$ is defined similarly;
see Figure~\ref{fig:pockets}.
\begin{figure}[b]
\begin{center}
\includegraphics{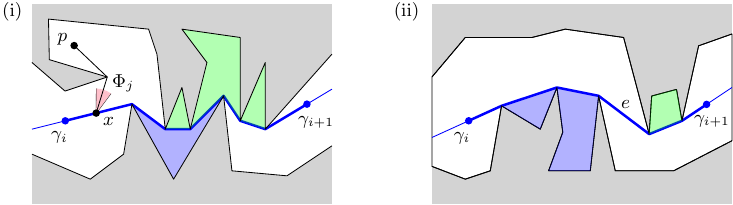}
\end{center}
\caption{Pockets in $\Hone$, which is the region below $\Gamma^*$, are shown in blue
and pockets in $\Htwo$ are shown in green. (i) The point $p$ is reachable from $x$
in direction of the cones $\Phi_j$, which is shown in pink. (ii) the edge $e$ is an intermediate edge.}
\label{fig:pockets}
\end{figure}
We refer to an edge of $\Gamma_i$ that belongs to a pocket of $\Gamma_i$ in 
$\Hone$ or $\Htwo$ as a \emph{pocket edge of $\Gamma_i$}. Any piece~$\Gamma_i$ 
has at most three non-pocket edges. Indeed, the first and last edge of~$\Gamma_i$ 
need not be pocket edges, since $\gamma_i$ and $\gamma_{i+1}$ need not be reflex vertices. 
Any other non-pocket edge must be an edge connecting two reflex vertices of $\P$,
a reflex vertex $u$ from $\Hone \cap \bd \P$ and
a reflex vertex~$v$ from $\Htwo \cap \bd\P$. Moreover, for $uv$ to be a non-pocket edge,
all vertices of $\Gamma_i$ that lie behind $v$ (as seen from $u$) must also 
be reflex vertices of $\Htwo \cap \bd\P$, and all vertices of $\Gamma_i$ that 
lie behind $u$ (as seen from $v$) must also 
be reflex vertices of $\Hone \cap \bd\P$. Hence, there can be at most one such edge,
which we will refer to as an \emph{intermediate edge of $\Gamma_i$}.
See Figure~\ref{fig:pockets}(ii) for an illustration.

For a point $p\in P$, a piece $\Gamma_i$ and a canonical cone $\Phi_j$, we will say that $p$ is \emph{reachable from $\Gamma_i$ in the cone $\Phi_j$}, if there exists a point $x\in \Gamma_i$ such that $\dir(x,p) \in \Phi_j$ and $\pi(x,p)\cap \Gamma_i = x$;
see Figure~\ref{fig:pockets}(i). In other words, there is a point $x\in \Gamma_i$ such that the shortest path $\pi(x,p)$ immediately leaves $\Gamma_i$ and  the direction of the first edge of $\pi(x,p)$ lies in the cone $\Phi_j$.  We may then just say that $p$ is reachable from $\Gamma_i$, or that $p$ is reachable from the edge of $\Gamma_i$ that $x$ belongs to.   Now we are ready to define the points placed in our coreset for each $\Gamma_i$.

\subparagraph{Coreset definition.}
For each piece $\Gamma_i$ we will put a set~$\coreset_i$ of points into the coreset.
The final coreset $\coreset$ is then defined as $\coreset=\bigcup_i \coreset_i$.
The points placed into $\coreset_i$ will be of five different types, as explained next.
\begin{itemize}
    \item If $\Gamma_i$ has an intermediate edge $e$, let $E_1,E_2$ denote the two half-polygons that $\P$ is split into by $e$. Then we start by placing in $\coreset_i$ the points $\alpha_i$, $\alpha'_i$ defined as follows: $\alpha_i$ is the point in $P\cap E_1$ which is furthest from $\gamma_i$ over all points $p\in P\cap E_1$. Similarly, $\alpha_i'$ is the point furthest from $\gamma_i$ over all points $p\in P\cap E_2$.
\end{itemize}
Note that for the points $\alpha_i,\alpha'_i$ we do not care whether they lie in $\Hone$ or~$\Htwo$.
We will also place in $\coreset_i$ some points of $P$ that must belong to $\Hone$,
and some points of $P$ that must belong to $\Htwo$. We denote these two sets of points by $\coreset_i(\Hone)$ and $\coreset_i(\Htwo)$ respectively.  We then have $\coreset_i=\coreset_i(\Hone)\cup \coreset_i(\Htwo) \cup\{ \alpha_i,\alpha'_i\}$. We now explain which points are placed in $\coreset_i(\Hone)$. The set $\coreset_i(\Htwo)$ is defined analogously. 
\begin{itemize}
    \item We place in $\coreset_i(\Hone)$ the point $r_i\in P\cap \Hone$ that is furthest from $\gamma_i$ over all points $p\in P\cap \Hone$.
    \item Let $X_i\subset P$ denote the points that lie in a pocket of $\Gamma_i$. Then we place in $\coreset_i(\Hone)$ the point $x_i\in X_i$ that is furthest from $\gamma_i$ over all points $p\in X_i\cap \Hone$.Note this is not done for each pocket separately: we only pick a single point over all pockets.
    \item Let $F_{i,j}$ denote the set of points that are reachable from $\first(\Gamma_i)$ in the cone $\Phi_j$. For each set $F_{i,j}$ we place in $\coreset_i(\Hone)$ the point $f_{i,j}$ defined as follows: $f_{i,j}$ is the point in $F_{i,j} \cap \Hone$ that is furthest away from $\gamma_i$.
    \item Let $L_{i,j}$ denote the set of points that are reachable from $\last(\Gamma_i)$ in the cone $\Phi_j$. For each set $L_{i,j}$ we place in $\coreset_i(\Hone)$ the point $\ell_{i,j}$ defined as follows: $\ell_{i,j}$ is the point in $L_{i,j} \cap \Hone$ that is furthest away from $\gamma_i$.
\end{itemize}   

\subparagraph{Proof of Theorem~\ref{thm:coreset} when the query point lies inside $\rch(P)$.}
Without loss of generality, we assume that $q \in \Htwo$.  
For brevity, we define $p := \fn(q)$
to be the furthest neighbor~$q$. The following lemma will used repeatedly in our proof.
\begin{lemma}\label{lem:cross-case}
Let $p^*\in \coreset$ be such that $\pi(q,p^*)$ intersects $\Gamma_i$,
and suppose $\|\pi(\gamma_i,p^*)\| \geq \|\pi(\gamma_i,p)\|$. 
Then $\|\pi(q,p^*)\|\geq (1-\varepsilon)\cdot \|\pi(q,p)\|$. Hence, $p^*$ is a good approximation for $\fn(q)$. 
\end{lemma}
\begin{proof}
Let $z$ be a point where $\pi(q,p^*)$ intersects $\Gamma_i$; see Figure~\ref{fig:cross-case}.
Then we have
\[
\begin{array}{lll}
 \|\pi(q,p)\|  & \leq & \|\pi(q,z)\|+\|\pi(z,\gamma_i)\| +\|\pi(\gamma_i,p)\| \hfill \mbox{(Triangle Inequality)} \\
               & \leq & \|\pi(q,z)\|+\|\pi(z,\gamma_i)\|+\|\pi(\gamma_i,p^*)\| \hfill \mbox{(by assumption)} \\
               & \leq & 2 \|\pi(z,\gamma_i)\| + \|\pi(q,z)\| + \|\pi(z,p^*)\| \hfill \mbox{(Triangle Inequality)} \\
               & \leq & 2\|\Gamma_i\|+ \|\pi(q,p^* )\| \hfill \mbox{(since $\|\pi(z,\gamma_i)\|\leq \|\Gamma_i\|$ and  $z\in\pi(q,p^*)$)} \\
               & < & \frac{\varepsilon}{2}\cdot\|\Gamma\|+\|\pi(q,p^*)\|  \hfill \mbox{(since $\|\Gamma_i\| \leq \frac{\eps}{6}\cdot\|\Gamma \|$}) \\
               & \leq & \varepsilon\cdot \|\pi(q,p)\|+\|\pi(q,p^*)\|  \hspace*{10mm} \mbox{(since $\|\Gamma\| \leq 2 \| \pi(q,\fn(q)) \|$ for any $q\in \P$)}
\end{array}
\]
and so indeed $\|\pi(q,p^*)\|\geq (1-\varepsilon)\cdot \|\pi(q,p)\|$.
\end{proof}
\begin{figure}
\begin{center}
\includegraphics{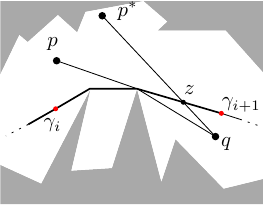}
\end{center}
\caption{Illustration for the proof of Lemma~\ref{lem:cross-case}.}
\label{fig:cross-case}
\end{figure}
If $p\in \{p_1,p_2\}$, then Theorem~\ref{thm:coreset} trivially holds since $\{p_1,p_2\}\subset \coreset$. 
So now we assume that $p\notin\{p_1,p_2\}$. Lemma \ref{lem:inside-rch} gives us that 
either $\pi(q,p)$ crosses $\Gamma^*$ or the backwards extension of $\pi(q,p)$ 
will intersect $\Gamma$. Note that in the former case $\pi(q,\fn(q))$ actually intersects $\Gamma$, 
since $q\in\rch(P)$ and shortest paths between points inside~$\rch(P)$ do not leave $\rch(P)$.
We first consider the latter case, where the backwards extension of $\pi(q,p)$ intersects~$\Gamma$.
\begin{lemma}\label{lem:first-case}
If the backwards extension of $\pi(q,p)$ intersects $\Gamma$, then 
there is a point $p^*\in\coreset$ such that $\|\pi(q,p^*)\|\geq (1-\varepsilon) \cdot \|\pi(q,p)\|$. 
\end{lemma}
\begin{proof}
Let $x$ be the point where the backwards extension of $\pi(q,p)$ hits~$\Gamma$ and let $\Gamma_i$ be the piece containing $x$.
Let
$r_i^*\in \coreset$ be the point reachable from $\Gamma_i$ which is furthest from $\gamma_i$ over all points $p\in P\cap \Htwo$
and that was put into the coreset~$\coreset$. Thus 
the definition of~$r_i$ implies $\|\pi(\gamma_i,r^*_i)\| \geq \|\pi(\gamma_i,p)\|$. 
We now have
\[
\begin{array}{lll}
 \|\pi(q,p)\| & = & \|\pi(x,p)\|-\|\pi(q,x)\|   \hfill \mbox{(since $x$ is on the backwards extension)}  \\
            & \leq & \|\pi(x,\gamma_i)\|+\|\pi(\gamma_i,p)\|-\|\pi(q,x)\| \hspace*{10mm} \mbox{(Triangle Inequality)}  \\
            & \leq & \|\pi(x,\gamma_i)\|+\|\pi(\gamma_i,r^*_i)\|-\|\pi(q,x)\|  \hfill \mbox{(definition of~$r^*_i$)} \\
            & \leq & 2 \|\pi(x,\gamma_i)\| +\|\pi(x,r^*_i)\|-\|\pi(x,q)\| \hfill \mbox{(Triangle Inequality)} \\
            & \leq & 2\|\Gamma_i\|+ \|\pi(x,r^*_i)\|-\|\pi(x,q)\| \hfill \mbox{(since $x\in\Gamma_i$ and $\gamma_i\in\Gamma_i$)}  \\
            &\leq & 2\|\Gamma_i\| +\|\pi(q,r^*_i)\| \hfill \mbox{(Triangle Inequality)} \\
            & < & \varepsilon\|\pi(q,p)\|+\|\pi(q,r^*_i)\| \hfill \mbox{(since $\|\Gamma_i\| \leq \frac{\eps}{6}\|\Gamma\| < \frac{\eps}{2} \|\pi(q,p)\| $)}
\end{array}
\]
and so $\|\pi(q,r^*_i)\|\geq (1-\varepsilon) \cdot \|\pi(q,p)\|$. By taking $p^* := r^*_i$
we thus obtain a point satisfying the conditions of the lemma.
\end{proof}
We are left with the case where $\pi(q,p)$ crosses $\Gamma_i$, which is handled by the next two lemmas. 
\begin{lemma}\label{lem:second-case}
If $\pi(q,p)$ crosses $\Gamma_i$ at a pocket edge or an intermediate edge, then 
there is a point $p^*\in\coreset$ such that $\|\pi(q,p^*)\|\geq (1-\varepsilon) \cdot \|\pi(q,p)\|$. 
\end{lemma}
\begin{proof}
Since we assumed that $q\in \Htwo$, we have $p\in \Hone$. 
If $\pi(q,p)$ crosses $\Gamma_i$ at a pocket edge,
then either $q$ or $p$ lies in a pocket of $\Gamma_i$. Recall that $\coreset$ contains the point $x_i\in X_i\cap \Hone$ 
that is furthest away from~$\gamma_i$, where $X_i$ denotes the set of points in a pocket of~$\Gamma_i$.
Note that $\pi(q,x_i)$ must cross $\Gamma_i$. If $p$ lies in a pocket then $\|\pi(\gamma_i,x_i)\| \geq \| \pi(\gamma_i,p)\|$ by definition of~$x_i$, and so we can apply Lemma~\ref{lem:cross-case} with~$p^*:=x_i$.
If $q$ lies in a pocket then $\|\pi(\gamma_i,r_i)\| \geq \| \pi(\gamma_i,p)\|$ by definition of~$r_i$, and so we can apply Lemma~\ref{lem:cross-case} with~$p^*:=r_i$.

If $\pi(q,p)$ crosses $\Gamma_i$ at its intermediate edge $e$, then for the point $\alpha_i\in \coreset$ 
it has to be that $\pi(q,\alpha_i)$ crosses~$e$. Hence, we can apply Lemma~\ref{lem:cross-case} with~$p^*:=\alpha_i$.
Thus in both cases the point $p^*$ has the desired property.
\end{proof}
It remains to deal with the most technically involved case, which is when $\pi(q,p)$ crosses 
$\first(\Gamma_i)$ or $\last(\Gamma_i)$. In the following proof, we denote by $\square_\pi abcd$ 
the geodesic quadrilateral with vertices $a,b,c,d \in \P$.
\begin{lemma}\label{lem:third-case}
    If $\pi(q,p)$ crosses $\first(\Gamma_i)$ or $\last(\Gamma_i)$ then 
    there is a point $p^*\in\coreset$ such that $\|\pi(q,p^*)\|\geq (1-\varepsilon) \cdot \|\pi(q,p)\|$. 
\end{lemma}
\begin{proof}
    Assume that $\pi(q,p)$ crosses $\first(\Gamma_i)$---the case that it crosses $\last(\Gamma_i)$ 
    is treated similarly---and let $x$ be the crossing point. Let $\Phi_j$ be the canonical cone
    such that $\dir(x,p)\in \Phi_j$. We will show that $f_{i,j}\in\coreset$ 
    is a good approximation of the furthest neighbor of $q$, so setting $p^* := f_{i,j}$ will prove the lemma. 
    Let $y\in\first(\Gamma_i)$ be a point from which $f_{i,j}$ is reachable. 
    We will consider three subcases.
    \begin{itemize}
    \item \emph{Subcase I: $\pi(q,f_{i,j})$ crosses $\Gamma_i$.} \\[2mm]
    By definition of~$f_{i,j}$ we have $\|\pi(\gamma_i,f_{i,j})\| \geq \|\pi(\gamma_i,p)\|$,
    and so it follows immediately from Lemma~\ref{lem:cross-case} that $\|\pi(q,f_{i,j})\|\geq (1-\varepsilon) \cdot \|\pi(q,p)\|$. 
    \end{itemize}
    \medskip 
    
    \noindent To define the other two subcases,
    let $q'$ denote the point where the paths $\pi(q,x)$ and $\pi(q,f_{i,j})$ split, and let $p'$ denote the 
    point where the paths $\pi(f_{i,j},q)$ and $\pi(f_{i,j},y)$ split; see Figure~\ref{fig:first_subcase2}.
    The two remaining subcases depend on whether $\dir(q',x)$ is very different from $\dir(x,p)$ or, 
    more precisely, depending on whether $\dir(q',x)$ lies in $\Phi_j$ or not. 
    
 \begin{figure}
    \begin{center}
    \includegraphics{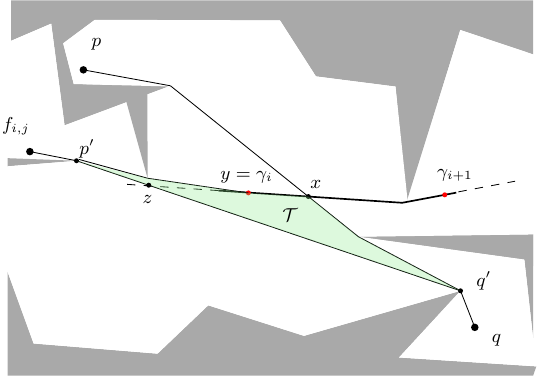}
    \end{center}
    \caption{Illustration for Subcase~II.}
    \label{fig:first_subcase2}
    \end{figure}

    \begin{figure}
    \begin{center}
    \includegraphics{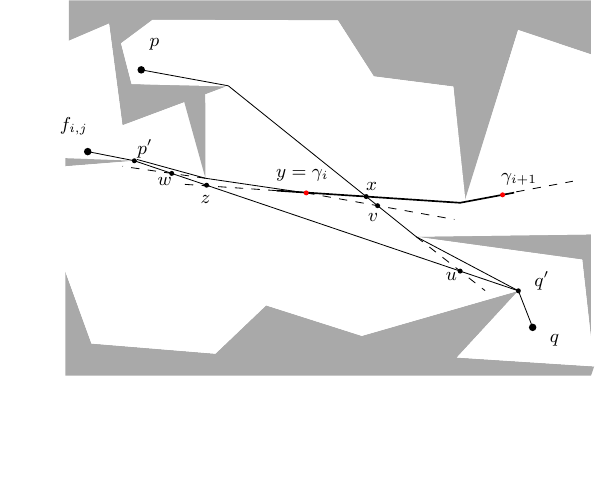}
    \end{center}
    \caption{Illustration for Subcase~II.}
    \label{fig:first_subcase1}
    \end{figure}

    \begin{itemize}
    \item \emph{Subcase II: $\pi(q,f_{i,j})$ does not cross $\Gamma_i$ and $\dir(q',x)\in \Phi_j$.} \\[2mm]
    Since also $\dir(x,p)\in \Phi_j$, we know that $\dir(e)\in\Phi_j$ 
    for every edge $e\in\pi(q',x)$. Let $z$ denote the point where $\pi(q,f_{i,j})$ crosses $\Gamma$, and 
    note that we assumed $z\not\in\Gamma_i$. An illustration is shown in Figure~\ref{fig:first_subcase2}, 
    where for ease of illustration we assumed that $y=\gamma_i$; this does not make any difference in the proof. 
  
    Consider the geodesic quadrilateral~$\square_\pi{q'xyp'}$. Then the four shortest paths forming the boundary
    of this quadrilateral are disjoint, except for shared endpoints. Since the
    quadrilateral is fully contained in~$\P$ (because $\P$ is simple), all interior vertices of
    these paths are reflex vertices of $\square_\pi{q'xyp'}$. 
    Moreover, the following holds: if $x$ lies between $z$ and $y$ on $\Gamma^*$ then 
    in ${\square_\pi}{q'xyp'}$ the angle at $x$ is reflex. Similarly, if $y$ lies between $z$ and $x$ on $\Gamma^*$ 
    (as in Figure~\ref{fig:first_subcase2})
    then the angle at $y$ is reflex. We will assume from now on without loss of generality
    that $y$ lies between $z$ and $x$. 

    We will first show that the edges of $\pi(q',p')$ all have a direction in the canonical cone~$\Phi_j$.
    \claiminproof{1}{For every edge $e\in \pi(q',p')$, we have that $\dir(e)\in \Phi_j$.}
    \begin{proofinproof}
    Since in ${\square_\pi}{q'xyp'}$ the angle at $y$ is reflex, ${\square_\pi}{q'xyp'}$ is
     a pseudo-triangle $\mathcal{T}$ with sides $\pi(q',x)$, and $\pi(x,y)\cup \pi(y,p')$,
    and $\pi(q',p')$. Let $e_1$ be the last edge of $\pi(q',x)$. Then we know that 
    $\dir(e_1)\in \Phi_j$ and $\dir(y,p')\in \Phi_j$.  Therefore, we have that $\mathbb{S}^1[\dir(e_1),\dir(y,p')]\subset \Phi_j$. 
    Since $\dir(e_1)\in\Span(q',x)$ and $\dir(y,p')\in\Span(x,p')$,
   Lemma \ref{lem:geo-triangle}(ii) gives us that $\Span(\pi(q',p'))\subset \mathbb{S}^1[\dir(e_1),\dir(y,p')] \subset \Phi_j$, which is what we wanted to prove.
    \end{proofinproof}
    Claim~1 implies the following.
    \claiminproof{2}{$\|\pi(q',x)\|+\|\pi(y,p')\|<(1+\frac{\varepsilon}{2})\|\pi(q',p')\| +\|\pi(x,y)\|$}
    \begin{proofinproof}
    We  extend $e_1$ to the interior of $\mathcal{T}$ 
    until it intersects $\pi(q',p')$ at a point $u$. We also extend $\first(\pi(y,p'))$ 
    to both directions until it intersects $\overline{xu}$ at a point $v$ 
    and $\pi(q',p')$ at a point $w$.
    Refer to Figure~\ref{fig:first_subcase1} for an illustration, where (to reduce the clutter in the figure) point $w$ coincides with $p'$.
        
    Then, by repeated application of the triangle inequality we have:
    \[
    \begin{array}{lll}
     \|\pi(q',x)\|+\|\pi(y,p')\| & \leq & \|\pi(q',u)\|+ \|\pi(u,x)\|+ \|\pi(y,p')\| \\
                & \leq & \|\pi(q',u)\|+ \|\overline{uv}\|+ \|\overline{vx}\| + \|\overline{yw}\|+\|\pi(w,p')\|\\
               & \leq & \|\pi(q',u)\|+\|\pi(w,p')\| + \|\overline{uv}\|+\|\overline{vy}\|+\|\pi(y,x)\| +\|\overline{yw}\| \\
              & = & \|\pi(q',u)\|+\|\pi(w,p')\| + \left( \|\overline{uv}\|+\|\overline{vw}\|\right)+\|\pi(x,y)\| 
    \end{array}
    \]
    In $\triangle_\pi{uvw}$, by Claim 1. we know that $\overline{uv},\overline{vw}$ are almost 
    parallel. More precisely, $\dir(\overline{uv})\in \Phi_j$ and $\dir(\overline{vw})\in \Phi_j$.
    Since our canonical cones have angle at most~$\eps/2$, this means that if we consider the (Euclidean) triangle $\triangle{uvw}$, 
    then the angles at $u$ and $w$ are both less than~$\varepsilon/2$. 
    This trivially gives us that 
    \[
    \|\overline{uv}\|+\|\overline{vw}\|<\left(1+\frac{\varepsilon}{2}\right)\|\overline{uw}\|<\left(1+\frac{\varepsilon}{2}\right)\|\pi(u,w)\|
    \]
    By combining this with our earlier derivation we obtain
    \begin{align*}
    \|\pi(q',x)\|+\|\pi(y,p')\|&< \|\pi(q',u)\|+\|\pi(w,p')\|+\left(1+\frac{\varepsilon}{2}\right)\|\pi(u,w)\|+\|\pi(x,y)\| \\ 
                         &<\left(1+\frac{\varepsilon}{2}\right)\|\pi(q',p')\| +\|\pi(x,y)\|,
    \end{align*}
    where the last inequality uses that $u,w\in \pi(q',p')$.
    \end{proofinproof}
    With Claim~2 in hand, we are ready to conclude the proof of Subcase~II: 
\[
\begin{array}{lll}
 \|\pi(q,p)\| & \leq & \|\pi(q,q')\| +\|\pi(q',x)\|+\|\pi(x,\gamma_i)\| +\|\pi(\gamma_i,p)\| \hspace*{5mm} \hfill \mbox{(Triangle Inequality)}\\
             & \leq & \|\pi(q,q')\| + \|\pi(q',x)\|+ \|\Gamma_i\| +\|\pi(\gamma_i,f_{i,j})\| \hfill \mbox{($x,\gamma_i\in\Gamma_i$ and definition of $f_{i,j}$)} \\
            & \leq & \|\pi(q,q')\| + \|\pi(q',x)\|+ \|\Gamma_i\| + \|\pi(\gamma_i,y)\| + \|\pi(y,f_{i,j})\| \hfill \mbox{(Triangle Inequality)} \\
             & \leq & \|\pi(q,q')\| + \|\pi(q',x)\|+ 2\|\Gamma_i\|  + \|\pi(y,f_{i,j})\|    \hfill \mbox{(since $y,\gamma_i\in \Gamma_i$)} \\
             & \leq & \|\pi(q,q')\| + \|\pi(q',x)\|+ 2\|\Gamma_i\| + \|\pi(y,p')\|   + \|\pi(p',f_{i,j})\|  \hfill \mbox{(Triangle Inequality)} \\
            & < & \|\pi(q,q')\|+2\|\Gamma_i\|+\left(1+\frac{\varepsilon}{2}\right)\|\pi(q',p')\| +\|\pi(x,y)\|   + \|\pi(p',f_{i,j})\|   \hspace*{3mm} \mbox{(by Claim~2)} \\
            & \leq &  3\|\Gamma_i\| + \|\pi(q,q')\| +\left(1+\frac{\varepsilon}{2}\right)\|\pi(q',p')\|   + \|\pi(p',f_{i,j})\| \hfill   \mbox{(since $x,y\in\Gamma_i$)} \\
            & \leq & 3\|\Gamma_i\|+ \left(1+\frac{\varepsilon}{2}\right)\|\pi(q,f_{i,j})\| 
                                          + \|\pi(p',f_{i,j})\|   \hfill \mbox{(since $q',p'\in\pi(q,f_{i,j}$))} \\
            & \leq & \frac{\varepsilon}{2}\|\pi(q,p)\|+\left(1+\frac{\varepsilon}{2}\right)\|\pi(q,f_{i,j})\| \hfill \mbox{(since $\|\Gamma_i\| \leq \frac{\eps}{6}\|\Gamma\| \leq \frac{\eps}{3} \|\pi(q,p)\| $)}
\end{array}
\]
and so we have 
\[
    \|\pi(q,f_{i,j})\| \ \ > \ \ \frac{1-\varepsilon/2}{1+\varepsilon/2} \cdot \|\pi(q,p)\|
                        \ \ > \ \ \left(1-\varepsilon\right) \cdot\|\pi(q,p)\|.
\]
\item \emph{Subcase III: $\pi(q,f_{i,j})$ does not cross $\Gamma_i$ and $\dir(q',x)\notin \Phi_j$.} \\[2mm]
In this case we also have $\dir(q',p')\notin \Phi_j$, because 
$\first(\pi(q',p'))$ lies further away from $\Phi_j$ than $\first(\pi(q',x))$ 
does.\footnote{This is true because as we walk back from $x$ to $q'$ along the reflex chain~$\pi(q',x)$, the direction turns monotonically. By the time
we reach $q'$, the direction is $\dir(q',x)$, which lies outside $\Phi_j$.
When we then turn to $\dir(q',p')$ we turn away from $\Phi_j$ even more.}
Below we show that this implies $\pi(q,p')$ crosses $\Gamma_i$.
But then $\pi(q,f_{i,j})$ crosses $\Gamma_i$, which contradicts the 
conditions of Subcase~III. Hence, Subcase~III cannot occur.

So, assume for a contradiction that $\pi(q,p')$ crosses $\Gamma^*$ at 
a point $z\notin \Gamma_i$, as shown schematically in Figure~\ref{fig:second_subcase}.
\begin{figure}
\begin{center}
\includegraphics{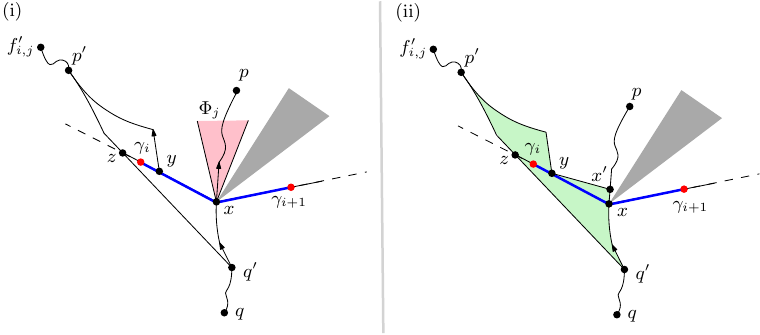}
\end{center}
\caption{ (i) A rough illustration for the case $\dir(q',x)\notin \Phi_j$. Note that in this example $x$ is a reflex vertex of $\P$. (ii) We can choose $x'$ close to $x$ on $\first(\pi(x,p))$ and repeat the proof of Claim~1 using ${\square_\pi}{q'x'yp'}$ (green region).}
\label{fig:second_subcase}
\end{figure}
As in Subcase~II we may assume that $y$ lies between $x$ and $z$ on $\Gamma^*$ and 
thus the angle at $y$ is reflex. Then essentially we can reprove Claim~1 of Subcase~II, 
which will give a contradiction since $\dir(q',p')\notin \Phi_j$. 
There is a technicality which lies in the case that $x$ coincides with a (reflex) vertex of $\P$. 
Then the direction of $e_1$, last edge  of $\pi(q',x)$,
need not be the same as the direction of the first edge of $\pi(x,p)$,
and so we may not have $\dir(e_1)\in  \Phi_j$.
(In Subcase~II, because $\dir(q',x)\in \Phi_j$,
we know that $\dir(e_1)\in  \Phi_j$ even when $v$ coincides with a reflex vertex of~$\P$.)
This can be bypassed with the following trick: choose $x'$ in $\first(\pi(x,p))$ very close to $x$ 
and observe that we can imitate the proof of Claim~1 using ${\square_\pi}{q'x'yp'}$ 
instead of ${\square_\pi}{q'xyp'}$. 
\end{itemize}
\end{proof}

\subsection{Extension to points outside $\rch(P)$}\label{outside-rch}
Lemma~\ref{lem:inside-rch} states that for any query point~$q\in\P$, either there is a point 
$q^*\in\Gamma$ with $\fn(q^*)=\fn(q)$ or $\pi(q,\fn(q))$ crosses~$\Gamma^*$. 
When $q\in\rch(P)$, then in the latter case we know that $\pi(q,\fn(q))$ crosses~$\Gamma$,
because for $q\in\rch(P)$ we have $\pi(q,\fn(q))\subset \rch(P)$. 
When $q\notin \rch(P)$ this need not be the case, however. To obtain a valid $\eps$-coreset
using the technique of the previous subsections,
it thus seems that we should subdivide $\Gamma^*$ (instead of~$\Gamma$) into pieces of length $\frac{\eps}{6} \|\Gamma\|$. This is problematic, however,
since $\Gamma^*$ can be arbitrarily much larger than $\Gamma$, and so the
size of our coreset may explode. Next we show how to handle this.

Recall that the points $p_1$ and $p_2$ define the diameter~$\Gamma$, and assume
without loss of generality that the segment $\overline{p_1 p_2}$ is horizontal. 
Let $\sigma$ be the square of side length $2\|\Gamma\|$ centered
at the midpoint of~$\overline{p_1 p_2}$. Then $\rch(P)$ lies completely inside~$\sigma$,
because for any point $r$ outside~$\sigma$ we have
$\max( \| \overline{p_1 r} \|, \| \overline{p_2 r} \| ) > \|\Gamma\|$.
We will now define a set $B$ consisting of $O(1)$ segments, each of length at most
$2\|\Gamma\|$, such that the following holds for any query point $q\not\in \rch(P)$:
the shortest path $\pi(q,\fn(q))$ intersects $\Gamma$ or it intersects a segment in~$B$.
The set $B$ is defined as follows; see Figure~\ref{fig:better-container}.
\begin{figure}
\begin{center}
\includegraphics{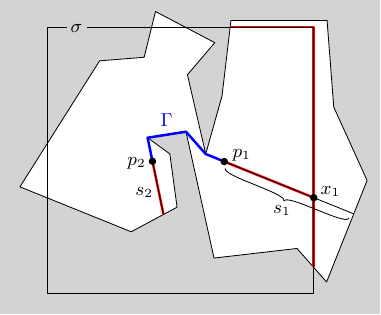}
\end{center}
\caption{The segments in $B$ are shown in red. Segment $s_1$ intersects
$\bd\sigma$, so besides $s_1\cap\sigma$, the segments forming the connected component
of $\bd\sigma \cap \P$ containing~$x_1$ are added to~$B$ as well.
Since $s_2$ lies fully inside $\sigma$,
no segments are added to $B$ for $s_2$, except $s_2$ itself. }
\label{fig:better-container}
\end{figure}

Consider the extension $\Gamma^* := \bfext(\Gamma)$. Note that $\Gamma^*\setminus \Gamma$
consists of two segments, a segment~$s_1$ incident to~$p_1$ and a segment $s_2$ incident to~$p_2$.
For $i\in \{1,2\}$, we add the following segments to~$B$.
We start by adding the segment $s_i\cap \sigma$ to~$B$.
Then, if $s_i$ intersects $\bd \sigma$, we add some more segments to $B$, as follows.
Let $x_i := s_i \cap \bd\sigma$, that is, $x_i$ is the point where $s_i$ exits~$\sigma$.
We add to $B$ the segments comprising the connected component of $\bd\sigma \cap \P$
that contains~$x_i$.
\begin{lemma} \label{lem:outside-rch}
Let $q\not\in\rch(P)$. Then $\pi(q,\fn(q))$ intersects $\Gamma$ or it intersects one of
the segments in~$B$.
\end{lemma}
\begin{proof}
By Lemma~\ref{lem:inside-rch} we know that $\pi(q,\fn(q))$ intersects $\Gamma^*$. If it
does so at a point of $\Gamma$ then we are done. If it does so at a point of
$s_1$ or $s_2$ inside $\sigma$ then we are done as well, since $s_1\cap \sigma$ 
and $s_2\cap \sigma$ were added to~$B$. Now suppose $\pi(q,\fn(q))$ intersects $\Gamma^*$
outside $\sigma$, say at a point $z\in s_i\setminus \sigma$. Note that $z$ is separated
from $\rch(P)$---and, hence, from $P$---by the connected component of $\bd\sigma \cap \P$
that contains~$x_i$. Hence, $\pi(q,\fn(q))$ must intersect one of the segments
of this connected component. Since we put these segments into~$B$, this finishes the proof.
\end{proof}
With Lemma~\ref{lem:outside-rch} available, it is straightforward to extend the coreset so that
it correctly handles query points $q\notin\rch(P)$. To this end we simply partition each of 
the segments of $B$ into $O(1/\eps)$ pieces of length at most $\frac{\eps}{6} \|\Gamma\|$,
and we apply the ``cone-technique'' to each such segment, that is, we treat them in the
same way as we treated the segments $\first(\Gamma_i)$ and $\last(\Gamma_i)$ in the
previous subsection. This is correct because in the proof of Lemma~\ref{lem:third-case}
we never used that $\Gamma_i \subset \Gamma$.  All we needed was that the piece $\Gamma_i$ 
is a shortest path, which trivially holds for the pieces we created for $B$, since each such piece is a segment.

%% file: 3-preprocessing.tex
\section{An efficient algorithm to construct the coreset}
In this section we explain how to construct the coreset~$\coreset$ defined in the 
previous section. We will show that we can do that quite efficiently, namely in time 
$O\left(\frac{1}{\eps} \left( n\log(1/\eps) + (n+m)\log(n+m)\right) \right)$. 

\subparagraph{Preprocessing.} 
Our approach crucially uses the diameter of $P$, so we start by computing $\Gamma$. 
This can be done in $O((m+n)\log(m+n)$ by first computing $\rch(P)$ in $O(m+n\log(m+n))$ 
time~\cite{toussaint-rch} and then computing its diameter in $O((m+n)\log(m+n))$ time~\cite{SURI1989220}.  
We then preprocess the subdivision of $\P$ induced by $\Gamma$ for point 
location~\cite{DBLP:reference/cg/Snoeyink04}, so that we can compute for each point 
in $P$ whether it lies in $\Hone$ or $\Htwo$, in $O(m+n\log m )$ time in total.

We now show how to compute the set $\coreset_i$ of points we put into our coreset~$\coreset$
for a given piece~$\Gamma_i$. We start by preprocessing~$\P$ for single-source shortest-path queries 
with $\gamma_i$ as the fixed source. This takes $O(m)$ time, after which we can compute the
distance from each point in~$P$ to $\gamma_i$ in $O(\log m)$ time~\cite{DBLP:journals/jcss/GuibasH89}.
Thus computing all distances takes $O(n\log m)$ time in total.
Recall that $\coreset_i$ contains the point furthest away from $\Gamma_{j-1}$ from each of the following sets.
\begin{itemize}
\item The sets $P\cap \Hone$ and $P\cap \Htwo$.
\item The sets $P\cap E_1$ and $P\cap E_2$,
      where $E_1$ and $E_2$ are the two half-polygons into which $\P$ is split by the intermediate edge~$e$
      of $\Gamma_i$, if such an edge exists.
\item The sets $X_i\cap \Hone$ and $X_i\cap \Htwo$, where $X_i\subset P$ contains the 
      points that lie in a pocket of~$\Gamma_i$.
\item The sets $F_{i,j}\cap \Hone$ and $L_{i,j}\cap \Hone$, and the set $F_{i,j}\cap \Htwo$ and $L_{i,j}\cap \Htwo$,
       where $F_{i,j}$ contains the points reachable from $\first(\Gamma_i)$ in the cone~$\Phi_j$, and
       $L_{i,j}$ contains the points reachable from $\last(\Gamma_i)$ in the cone~$\Phi_j$.
\end{itemize}
The sets $P\cap E_1$, $P\cap E_2$, $X_i\cap \Hone$, $X_i\cap \Htwo$ can each be determined 
by performing a point-location query for each point $p\in P$  in a suitably defined subdivision 
of complexity~$m$. This can be done in $O((m+n)\log{m})$ time in total~\cite{DBLP:reference/cg/Snoeyink04}. 
Since we precomputed all distances to $\gamma_i$, this also gives us the furthest neighbor of $\gamma_i$ 
from each of these sets.   So far we spent $O((n+m)\log m)$ time in total to compute~$\coreset_i$.
\medskip

Efficiently determining the furthest neighbors in the sets  $F_{i,j}\cap \Hone$, $F_{i,j}\cap \Htwo$, $L_{i,j}\cap \Hone$, and  $L_{i,j}\cap \Htwo$
is slightly more complicated. Next we explain how to handle $F_{i,j}\cap \Hone$;
the other sets can be computed analogously. 

Consider a point $p\in P$. Recall that we preprocessed~$\P$ for single-source shortest-path queries 
with $\gamma_i$ as the fixed source. Using this data structure, we can not only
find $\|\pi(\gamma_i,p)\|$ in $O(\log m)$ time, but we can also report the first edge
of this path in the same amount of time.\footnote{Indeed, shortest-path queries are answered
using a shortest-path tree with the source node as root. It then suffices to know, for each
node~$v$ in the shortest-path tree, which subtree of the root contains~$v$.}
Let $f_i$ be the endpoint of $\first(\Gamma_i)$
other than $\gamma_i$. We also preprocess~$\P$ for single-source shortest-path queries 
with $f_{i}$ as the fixed source, so that we can find $\first(\pi(f_i,p))$ in $O(\log n)$ time.
After computing these edges, we know all directions $\phi$ such that there is a point $x\in\first(\Gamma_i)$
with $\dir(x,p)=\phi$, since these are simply the directions in between $\dir(\gamma_i,p)$
and $\dir(f_{i},p)$; see Figure~\ref{fig:preprocessing}.  
\begin{figure}
\begin{center}
\includegraphics{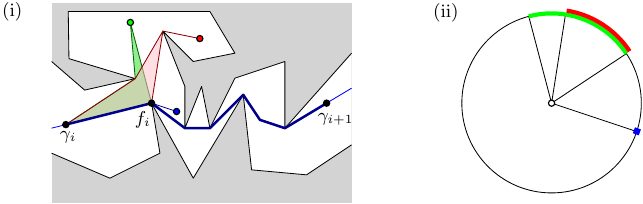}
\end{center}
\caption{(i) The shortest paths from three points, shown in green, red, and blue, to $\gamma_i$
and to~$f_i$. (ii) The directions from which these points can be reached are the directions in between
$\dir(\gamma_i,p)$ and $\dir(f_{i},p)$. Note that for the blue point these directions
are the same.}
\label{fig:preprocessing}
\end{figure}
This is true since, as $x$ moves 
along $\first(\Gamma_i)$ from $\gamma_i$ to $f_i$, the direction of the first edge changes continuously
from $\dir(,p)$ to $\dir(f_{i},p)$.
Hence, we can list in $O(\frac{1}{\eps} + \log m)$ time for a given point $p\in P\cap \Hone$ all
cone directions for which $p$ is relevant. Thus we can compute all sets
$F_{i,j}\cap \Hone$ in $O(n/\eps + n\log m)$ time in total.

We can speed this up using standard range-searching techniques. To this end we build a 
balanced binary search tree~$\T_i$ over the canonical cone directions, seen as sub-intervals 
of~$[0,360^{\circ}]$. Because of the way the cones are generated, this can be done in
linear time. Each point $p\in P$ induces a set $I(p)$ of 
one or two sub-intervals of $[0,360^{\circ}]$, corresponding to the directions in which
$p$ can be reached from $\first(\Gamma_i)$.
We assign the distance of $p$ to $\gamma_i$ as a weight to $I(p)$. 
Then we want to know, for each cone~$\Phi_j$, the maximum weight of any $I(p)$
that such that $\Phi_j\cap I(p)\neq\emptyset$. To compute this we perform a range query
with each~$I(p)$, selecting a set $N(p)$ of $O(\log(1/\eps))$ nodes in $\T_i$ such that 
$\Phi_j\cap I(p)\neq\emptyset$ if and only if $\Phi_j$ corresponds to a leaf in
a subtree rooted at a node in $N(p)$. Then we determine for each node~$v$ 
in $\tree_i$ the maximum weight of any $I(p)$ such that $v\in N(p)$,
taking $O(n\log(1/\eps))$ time in total. Finally, we traverse $\tree$ in a top-down manner 
to compute for each of its leaves the largest weight associated with any of the nodes on its search path. Thus finding the furthest neighbor from $\gamma_i$
in $F_{i,j}\cap \Hone$, over all $j$, takes $O(n\log(1/\eps))$ time in total.

To sum up, for each $\Gamma_i$ we can compute $\coreset_i$ in time $O(n\log(1/\eps) + (m+n)\log m)$. 
Since the above computations need to be done for each piece $\Gamma_i$ and there are $O(1/\eps)$ pieces, 
the coreset $\coreset$ can be constructed in $O\left(\frac{1}{\eps} \left( n\log(1/\eps) + (n+m)\log(n+m)\right) \right)$ time in total. 
(Actually, the bound is $O\left( (n+m)\log(n+m) + \frac{1}{\eps} \left( n\log(1/\eps) + (n+m)\log m\right) \right)$, but we prefer the slightly shorter expression given above.)
This concludes the proof of Theorem~\ref{thm:coreset}.

%% file: 4-conclusion.tex
\section{Concluding remarks}
We proved that any set $P$ of $n$ points in a simple polygon~$\P$ of complexity~$m$
admits an $\eps$-coreset of size $O(1/\eps^2)$ for furthest-neighbor queries. 
Thus the size of the coreset is independent of $n$ and $m$. This immediately
gives a data structure for approximate furthest-neighbor queries whose size is independent
of~$n$.  The size of the data structure will be linear in $m$,
but this is unavoidable: it is easy to see that for $n=2$ we already
have an $\Omega(m)$ lower bound on the storage of any data structure for 
approximate furthest-neighbor queries.

Another application, which we did not mention so far, is the $1$-center
problem with $z$~outliers in a simple polygon, for which we can now obtain an
$\eps$-coreset of size $O(z/\eps^2)$. To this end we use the 
technique of Agarwal, Har-Peled, and Yu~\cite{DBLP:conf/soda/AgarwalHY06}:
compute an $\eps$-coreset~$\coreset$ for furthest-neighbor queries on~$P$, 
remove $\coreset$ from $P$, compute an $\eps$-coreset for furthest-neighbor queries 
on the remaining points and add these to the first coreset, and so on. If this is repeated $z+1$ times, the result is an $\eps$-coreset for the 1-center problem with $z$ outliers.

There are several directions for further research. First of all, we do not expect our
bound to be tight: we believe that an $\eps$-coreset of size $O(1/\eps)$, or perhaps even $O(\sqrt{1/\eps})$, should exist. One approach to reduce the size of the coreset is
to prove this may be to show that 
in our method we can re-use many points, that is, that the $O(1/\eps)$ subsets $\coreset_i$
that we define for each of the $O(1/\eps)$ pieces $\Gamma_i$ will have (or can be made to have)
many points in common.  
Another challenging open problem is to extend our
result to polygons with holes. This seems to require a different approach, as 
our reduction to points on the diameter fails in this setting.

%% file: app-triangle-proof.tex
\subsection{Proof of Lemma~\ref{lem:geo-triangle}}

\geotriangle*

\begin{proof}
\begin{enumerate}[(i)]
\item Without loss of generality, assume that $\pi(p,q)$ and $ \pi(q,r)$ contain a common direction $d_1\in \mathbb{S}^1$ or two antipodal directions $d_2, d^*_2\in \mathbb{S}^1$. In either case, we can assume that the directions are horizontal. Thus we can take a horizontal tangent $\ell_1$  at $p_1\in \pi(p,q)$ and a horizontal tangent $\ell_2$ at $p_2\in \pi(q,r)$, such that $\ell_1,\ell_2$ are either both below or above $\pi(p,q),\pi(q,r)$ respectively. Assume that they are both below and that $\ell_1$ is lower than $\ell_2$, as in Figure~\ref{fig:s1-directions}(i). Then $\ell_1$ must exit $\triangle_\pi{pqr}$ via $\pi(p,r)$ at two different points $u,v$. Since $\pi(p,r)$ is a shortest path, this means that $\overline{uv}= \pi(u,v)$. But then $\pi(p,r)$ must go through $p_1$ which contradicts the fact that $\pi(p,q)$ and $ \pi(p,r)$ are internally disjoint (that is, they only intersect at a shared endpoint).
\item Let $\Span(\pi(p,q))= \mathbb{S}^1[a_1,a_2]$, let  $\Span(\pi(q,r))=\mathbb{S}^1[b_1,b_2]$, and let $\Span(\pi(r,p))=\mathbb{S}^1[c_1,c_2]$. Property i) implies that $\mathbb{S}^1[b_1,b_2] \cap \left( \mathbb{S}^1[a_1,a_2] \cup \mathbb{S}^1[a^*_1,a^*_2] \right) =\emptyset$, and so $\mathbb{S}^1[b_1,b_2]\subset \mathbb{S}^1[a_1,a^*_2]$ or $\mathbb{S}^1[b_1,b_2]\subset\mathbb{S}^1[a^*_1,a_2]$. Without loss of generality, assume that $\mathbb{S}^1[b_1,b_2]\subset\mathbb{S}^1[a^*_1,a_2]$, as in Figure~\ref{fig:s1-directions}(ii) .  Now consider $\triangle pqr$ (refer to Figure~\ref{fig:s1-directions}(iii)) and let $\ell=\dir(\overline{pq})$ let $j=\dir(\overline{qr})$ and let $k=\dir(\overline{rp})$. Observe that $\ell\in \mathbb{S}^1[a_1,a_2]$, $j\in \mathbb{S}^1[b_1,b_2]$ and $k\in \mathbb{S}^1[c_1,c_2]$. Moreover, the directions $\ell,j,k$ must span the plane, that is, they cannot lie in a semi-circle of $\mathbb{S}^1$: each of the arcs $\mathbb{S}^1[\ell. j]$, $\mathbb{S}^1[j,k]$, $\mathbb{S}^1[\ell,k]$ corresponds to an exterior angle of $\triangle pqr$ and the exterior angles add up to $2\pi$. This would not possible if $\ell,j,k$ lied in a semi-circle of $\mathbb{S}^1$. Therefore $\mathbb{S}^1[c_1,c_2] \subset \mathbb{S}^1[a^*_2,b^*_1]$, which implies that $\Span(\pi(p,r))=\mathbb{S}^1[c^*_1,c^*_2]\subset \mathbb{S}^1[a_2,b_1]\subset \mathbb{S}^1[\phi_1,\phi_2]$.
\end{enumerate}
\end{proof}
\begin{figure}[h]
\begin{center}
\includegraphics{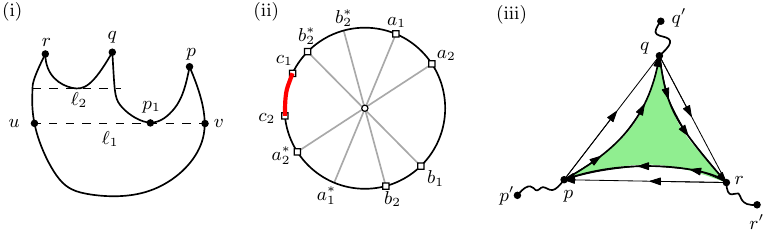}
\end{center}
\caption{Illustrations for the proof of Lemma~\ref{lem:geo-triangle}.}
\label{fig:s1-directions}
\end{figure}